\documentclass[sigconf,nonacm]{acmart}

\usepackage{dsfont}
\usepackage{algorithm}
\usepackage{algpseudocode}
\usepackage{subfig}
\usepackage{stfloats}
\newtheorem{assumption}{Assumption}
\newtheorem{remark}{Remark}
\definecolor{brickred}{cmyk}{0,0.89,0.94,0.28}
\definecolor{goldenrod}{cmyk}{0,0.10,0.84,0}
\definecolor{purple}{cmyk}{0.45,0.86,0,0}
\definecolor{rawsienna}{cmyk}{0,0.72,1,0.45}
\definecolor{olivegreen}{cmyk}{0.64,0,0.95,0.40}
\definecolor{peach}{cmyk}{0,0.5,0.7,0}
\definecolor{darkolive}{rgb}{0.,0.4,0.}
\colorlet{grey}{gray!40}

\def\tcbr{\textcolor{brickred}}

\AtBeginDocument{%
  }
\setcopyright{none} 
\acmDOI{XXXXXXX.XXXXXXX}

\makeatletter
\renewcommand\@formatdoi[1]{\ignorespaces}
\makeatother
\acmConference[xxxx]{}{xxxx}{xxxx}

\acmISBN{xxx-x-xxxx-XXXX-X/xxxx/xx}

\begin{document}

\title{$(\ell,\delta)$-Diversity: Linkage-Robustness via a Composition Theorem}

\author{V.~Arvind Rameshwar}
\affiliation{%
  \institution{India Urban Data Exchange Program Unit}
  \city{Indian Institute of Science, Bengaluru}
  \country{India}
}
\authornote{Both authors contributed equally to this research.}
\email{arvind.rameshwar@gmail.com}
\author{Anshoo Tandon}
\authornotemark[1]
\affiliation{%
  \institution{India Urban Data Exchange Program Unit}
  \city{Indian Institute of Science, Bengaluru}
  \country{India}
}
\email{anshoo.tandon@gmail.com}








\renewcommand{\shortauthors}{Rameshwar and Tandon}

\begin{abstract}
In this paper, we consider the problem of degradation of anonymity upon linkages of anonymized datasets. We work in the setting where an adversary links together $t\geq 2$ anonymized datasets in which a user of interest participates, based on the user's known quasi-identifiers, which motivates the use of $\ell$-diversity as the notion of dataset anonymity.
We first argue that in the worst case, such linkage attacks can reveal the \emph{exact} sensitive attribute of the user, even when each dataset respects $\ell$-diversity, for moderately {large} values of $\ell$. This issue motivates our definition of (approximate)  $(\ell,\delta)$-diversity -- a parallel of (approximate) $(\epsilon,\delta)$-differential privacy (DP) -- which simply requires that a dataset respect $\ell$-diversity, with high probability. We then present a mechanism for achieving $(\ell,\delta)$-diversity, in the setting of independent and identically distributed samples. Next, we establish bounds on the degradation of $(\ell,\delta)$-diversity, via a simple ``composition theorem,'' similar in spirit to those in the DP literature, thereby showing that approximate diversity, unlike standard diversity, is roughly preserved upon linkage. Finally, we describe simple algorithms for maximizing utility, measured in terms of the number of anonymized ``equivalence classes,'' and derive explicit lower bounds on the utility, for special sample distributions.
\end{abstract}





\maketitle

\section{Introduction}
The past few decades have seen an explosion in the quantity of data being collected, stored, and disseminated. This tremendous increase in data availability has predictably raised concerns about the privacy of the data samples of individual users. Several works \cite{sweeney,dinurnissim,narayanan,homer,taxi} have demonstrated that the release of even simple functions of a dataset that is not publicly available can leak sensitive information about the identities of users present in the datasets. In many settings of interest, data owners such as hospitals may be incentivized to share \emph{microdata}, via a release of \emph{entire datasets} with clients who desire to carry out statistical analyses on the data records. 

Clearly, in such settings, a simple suppression of user identities is insufficient for protecting the sensitive information of users, since an adversary can make use of \emph{quasi-identifiers} present in records contributed by a user, to obtain his/her sensitive attribute. Towards alleviating data leakage via such reconstruction attacks, several notions of data anonymity were introduced \cite{sweeney,k-anon1,k-anon2,k-anon3,l-div,t-close}. Popular such notions of anonymity include $k$-anonymity \cite{sweeney,k-anon1,k-anon2,k-anon3}, $\ell$-diversity \cite{l-div}, and $t$-closeness \cite{t-close}. Informally, $k$-anonymity enforces the condition that for every anonymized quasi-identifier (or ``equivalence class'') in a dataset, there must exist at least $k\geq 2$ data records. Likewise, $\ell$-diversity mandates that there must be at least $\ell$ \emph{distinct} sensitive attributes in an equivalence class, and $t$-closeness requires that the distribution of sensitive attributes in any equivalence class be close (measured by $t$) to that in the whole anonymized dataset. Typically, such anonymity notions are enforced by \emph{suppressing}, or dropping, selected quasi-identifiers, or by \emph{generalizing}, or binning, quasi-identifiers together to form so-called \emph{equivalence classes} that partition the space of quasi-identifiers. An example of a dataset that is anonymized (via generalization) to respect $2$-anonymity and $2$-diversity is shown in Table \ref{tab:anon}; the sensitive attribute here is the disease each user suffers from. We refer the reader to several works such as \cite{complexity-k-anon,k-anon-approx-1,k-anon-approx-2,k-anon-approx-3, k-anon-heuristic-1,mondrian} and references therein for results on the computational complexity of, and approximation algorithms and heuristics for, $k$-anonymity.

\begin{table}
	\centering
	\begin{tabular}{||c|c|c||}
		\hline
		Gender & Postal Code & \tcbr{Disease}\\
		\hline
		\hline
		Male & 560012 & \tcbr{Heart disease}\\
		Male & 560011 & \tcbr{Lung infection}\\
		Female & 560010 & \tcbr{Osteoporosis}\\
		Female & 560010 & \tcbr{Cervical cancer}\\
		Female & 560010 & \tcbr{Osteoporosis}\\
		\hline
	\end{tabular}
	\caption*{(a) Raw dataset}
\vspace{1em}
	\begin{tabular}{||c|c||}
	\hline
	(Gender, Postal Code) & \tcbr{Disease}\\
	\hline
	\hline
	(Male, 560012) or (Male, 560011) & \begin{tabular}{@{}c@{}}\tcbr{Heart disease} \\ \tcbr{Lung infection}\end{tabular}\\
	\hline
	(Female, 560010) & \begin{tabular}{@{}c@{}}\tcbr{Osteoporosis}\\ \tcbr{Osteoporosis}\\ \tcbr{Cervical cancer}\end{tabular}\\
	\hline
\end{tabular}
\caption*{(b) Anonymized dataset}
\caption{An illustration of a dataset that has been anonymized via generalization of quasi-identifiers, to respect $2$-anonymity and $2$-diversity}
\label{tab:anon}
\end{table}

As an illustration of the frailities of such standard anonymity notions, consider the instance of two hospitals in sufficiently close neighbourhoods that store medical records of patients who seek its diagnostic services.
Each hospital \emph{independently} carries out generalization or suppresion of quasi-identifiers of records, in an attempt to maintain patient anonymity, during its releases to clients for analyses. Now, consider an adversary (or malicious client) that obtains the \emph{anonymized} datasets from \emph{both} hospitals on a fixed day. The adversary, in addition, is aware of the quasi-identifiers of a specific user who lies in both datasets, and seeks to extract the sensitive attribute of this user. Clearly, an appropriate notion of anonymity in the face of such a threat model is $\ell$-diversity; indeed, a \emph{single} dataset that respects $k$-anonymity may directly reveal the sensitive attribute of the user, if all sensitive attributes in the equivalence class of interest are equal. {It can be seen that a similar drawback precludes the use of $t$-closeness as well, since datasets wherein each equivalence class contains (potentially different counts of) the \emph{same}, \emph{common} sensitive attribute value, obeys $0$-closeness, but this sensitive attribute value can immediately be extracted from such an ``anonymized'' dataset.} We mention that another instance where attacks of this nature are possible is the setting of datasets of voter records across constituencies, collected by several different polling/news agencies (see \cite{loksabha} for examples of such datasets), where the sensitive attribute of a user record is the candidate the user voted for.

In this work, we first establish, via rigorous arguments, the worst-case degradation of  $\ell$-diversity upon the ``linkage''\footnote{The term ``linkage'' used in this paper has a slightly different interpretation compared to the usage in previous work \cite{k-anon1,k-anon3}, in that prior literature focussed on the adversary's ability to ``link'' user identities (available via public datasets) to sensitive information (present in non-public datasets). In this work, our focus is on the setting where a user contributes \emph{identical} records to two different datasets, thereby compromising its sensitive attributes, when the datasets are joined together based on the common, known quasi-identifiers.} of anonymized datasets. Importantly, we demonstrate that the anonymity parameters can degrade enormously upon linkage, in the worst-case, even when the original parameters are moderately large -- see Table \ref{tab:linkage} for an example of such a linkage attack.
Observe that each of $\overline{\mathcal{D}}_1$ and $\overline{\mathcal{D}}_2$ respects both $2$-diversity, but an adversary that is in possession of both anonymized datasets and the quasi-identifier (Female, 560010) of a user of interest, who lies in both datasets, can directly obtain the disease "cervical cancer" that she suffers from.

 This then motivates our definition of \emph{approximate} $(\ell,\delta)$-diversity, which requires that $\ell$-diversity be satisfied with probabililty at least $1-\delta$, for a fixed $\delta\in (0,1)$. We mention that our definition of approximate anonymity is in the same vein as that of (approximate) $(\epsilon,\delta)$-differential privacy (DP) (see, e.g., \cite[Def. 2.4]{dworkroth}, \cite[Def. 1.4]{vadhan2017}), which can be interpreted as enforcing $\epsilon$-DP with probability at least $1-\delta$.\footnote{We mention in this context that there also exists literature \cite{k-anon-dp} on pre-processing datasets via a sampling procedure, to ensure that a $k$-anonymization algorithm performed thereafter also satisfies DP.} Furthermore, we note that our definition of $(\ell,\delta)$-diversity, introduced in the context of dataset linkages, is in spirit similar to the definition of ``probabilistic $k$-anonymity'' introduced in \cite{prob-anon}, which also allows for $k$-anonymity failure with probability at most $\beta$, for some $\beta>0$ (see Definition 4 therein).

With our definition in place, we establish a simple mechanism (or algorithm) for achieving $(\ell,\delta)$-diversity in a \emph{sample-independent} manner, for the setting when the samples are independent and identically distributed (i.i.d.), via a bound on the size of the dataset required for $(\ell,\delta)$-diversity. We mention that a dataset size bound that is similar in spirit was derived for the setting of ``probabilistic $k$-anonymity,'' in \cite[Thm. 4]{prob-anon}; however, the analysis therein was not motivated by the question of dataset linkages, and their work does not provide any \emph{explicit} algorithms for achieving ``probabilistic $k$-anonymity,'' or discuss the linkage-resilience properties of the anonymity notion. In our work, we show via a simple \emph{composition theorem} (similar to those well-known in the DP literature, see, e.g., \cite[Sec. 3.5]{dworkroth}), that for the case of i.i.d. samples, $(\ell,\delta)$-diversity thus achieved is roughly preserved upon dataset linkage, unlike ``pure" $\ell$-diversity. Finally, we discuss a simple, yet optimal, algorithm for maximizing the utility of our mechanism for achieving $(\ell,\delta)$-diversity, in the natural setting where we require equivalence classes to contain \emph{contiguous} quasi-identifiers, according to some fixed, total order on quasi-identifiers. We mention that in this work, we measure the utility of an anonymization mechanism as simply as the number of equivalence classes constructed. We then present explicit lower bounds on the utility, for special sample distributions where the sensitive attributes are independent of the quasi-identifiers.

The paper is organized as follows: Section \ref{sec:prelim} discusses some background on anonymity and sets down the problem formulation; Section \ref{sec:worst-case} argues that standard $\ell$-diversity can degrade enormously upon linkage, in the worst-case; Section \ref{sec:approx} presents the definition of (approximate) $(\ell,\delta)$-diversity, presents an explicit mechanism for i.i.d. samples, and discusses generalization strategies for maximizing utility and lower bounds thereon; and Section \ref{sec:numerics} provides experiments on the number of samples required for guaranteeing $(\ell,\delta)$-diversity via our mechanism, for a setting of practical interest. The paper is then concluded in Section \ref{sec:conclusion}, with some directions for future research.
	\begin{table}
		\centering
		\begin{tabular}{||c|c||}
		\hline
		(Gender, Postal Code) & \tcbr{Disease}\\
		\hline
		\hline
		(Male, 560012) or (Male, 560011) & \begin{tabular}{@{}c@{}}\tcbr{Heart disease} \\ \tcbr{Lung infection}\end{tabular}\\
		\hline
		(Female, 560010) & \begin{tabular}{@{}c@{}}\tcbr{Osteoporosis}\\  \tcbr{Cervical cancer}\end{tabular}\\
		\hline
	\end{tabular}
	\caption*{Anonymized dataset $\overline{\mathcal{D}}_1$}
	\vspace{1em}
	\begin{tabular}{||c|c||}
		\hline
		(Gender, Postal Code) & \tcbr{Disease}\\
		\hline
		\hline
		(Female, 560008) & \begin{tabular}{@{}c@{}}\tcbr{Irritable bowel syndrome} \\ \tcbr{Lung infection}\end{tabular}\\
		\hline
		(Female, 560010) & \begin{tabular}{@{}c@{}}\tcbr{Cervical cancer} \\ \tcbr{Amoebic dysentry}\end{tabular}\\
		\hline
	\end{tabular}
\caption*{Anonymized dataset $\overline{\mathcal{D}}_2$}
\caption{Anonymized datasets that can be linked based on the equivalence class (Female, 560010) to reveal that the \emph{only} sensitive attribute with this equivalence class is the disease \tcbr{cervical cancer}}
\label{tab:linkage}
	\end{table}

\section{Notation and Preliminaries}
\label{sec:prelim}
\subsection{Notation}
The notation $\mathbb{N}$ denotes the set of positive natural numbers. For $n\in \mathbb{N}$, the notation $[n]$ denotes the set $\{1,2,\ldots,n\}$. Random variables are denoted by upper-case letters, e.g., $X, Y$, and their realizations by small letters, e.g., $x,y$. For a discrete distribution $P$ on an alphabet $\mathcal{X}$, we define its support supp$(P):= \{x\in \mathcal{X}: P(x)>0\}$. Given a real-valued vector $\mathbf{a}= (a_1,a_2,\ldots,a_n)$, we denote by $w(\mathbf{a})$ its (Hamming) weight, i.e., $w(\mathbf{a})$ is the number of non-zero coordinates in $\mathbf{a}$. The notation $\mathbf{1}$ denotes the all-ones vector, whose length can be inferred from the context. The indicator function $\mathds{1}\{A\}$ equals $1$, if the property $A$ is true, and $0$, otherwise. For positive integers $a, n$, we define {$(a)_n:= \text{mod}(a-1,n)+1$}, where mod$(a,n)$ returns the value $r\in \{0,1,\ldots,n-1\}$ such that $a = tn+r$, for some $t\in \mathbb{N}$. We use `$\ln$' to refer to the natural logarithm. Given sequences $(a_n)_{n\geq 1}, (b_n)_{n\geq 1}$ of positive reals, we say that $a_n = O(b_n)$ if there exists $C>0$ such that for large enough $n$, we have $a_n\leq C\cdot b_n$.

\subsection{Anonymity: Formal Definitions}
\label{sec:background}
Consider a dataset $\mathcal{D}$ consisting of records of user identities, non-sensitive quasi-identifiers, and sensitive attributes. More precisely, $\mathcal{D}$ is the collection $\{(i,q_{1,i},q_{2,i},\ldots,q_{m,i},s_i):\ i \in [N]\}$, where $N$ is the total number of users and each user $i \in [N]$ contributes quasi-identifiers $\mathbf{q}_i:=(q_{1,i},q_{2,i},\ldots,q_{m,i})\in \mathcal{Q}$, for some fixed $m\geq 1$. Further, $s_i \in \mathcal{S}$ is the sensitive attribute corresponding to user $i$. We assume that the alphabets $\mathcal{Q}$ and $\mathcal{S}$ of quasi-identifiers and sensitive attributes, respectively, are finite and independent of $N$. We first recollect the definition of $\ell$-diversity. For $s\in \mathcal{S}$, let $n^{(\mathbf{q})}(s):=|\{i\in [N]:\ \mathbf{q}_i = \mathbf{q}\text{ and } s_i = s\}$.


\begin{definition}[$\ell$-diversity \cite{l-div}]
	\label{def:l-div}
	A dataset $\mathcal{D}$ satisfies $\ell$-diversity, if $\ell\in [|\mathcal{S}|]$ is the largest integer such that for any vector $\mathbf{q}\in \mathcal{Q}$ with $\sum_{s \in \mathcal{S}} n^{(\mathbf{q})}(s)>0$, we have 
	\[
	|\{s\in \mathcal{S}:\ n^{(\mathbf{q})}(s)>0\}|\geq \ell.
	\]
\end{definition}
In fact, following arguments made in \cite{l-div}, one can show that $\ell$-diversity arises as a natural notion of anonymity from operational, information-theoretic arguments about the optimal strategy of an adversary, which possesses \emph{limited} information---only the unique sensitive attribute values corresponding to each quasi identifier $\mathbf{q}'$---about the dataset $\mathcal{D}$ and a quasi-identifier $\mathbf{q}$ of a user whose sensitive attribute he/she wishes to determine. 
%

In practice, $\ell$-diversity is achieved by first dropping the user identities, and then carrying out either \emph{suppression}, or dropping of selected quasi-identifiers of selected users, or via \emph{generalization}, or combining/binning of quasi-identifiers. In this work, we restrict our attention to $\ell$-diversity achieved via generalization.

\subsection{Problem Formulation}
\label{sec:problem}
Consider datasets $\mathcal{D}_1, \mathcal{D}_2, \ldots, \mathcal{D}_t$ of equal size $N$, that have been anonymized to obey $\ell$-diversity, giving rise to anonymized datasets $\overline{\mathcal{D}}_1,\overline{\mathcal{D}}_2,\ldots,\overline{\mathcal{D}}_t$. We make the following assumption on the users present in the datasets:
\begin{assumption}
	We assume that at least $1$ user is common among the datasets $\left\{\mathcal{D}_j: j\in [t]\right\}$.
\end{assumption}
Observe that in the absence of user indentities, a sufficient statistic for the representation of an anonymized dataset $\overline{\mathcal{D}}$ is simply the collection of counts $\{n^{(\overline{\mathbf{q}})}(s)\}$ corresponding to an \emph{anonymized} quasi-identifier $\overline{\mathbf{q}}$ (or ``equivalence class''), obtained potentially via generalization. In what follows, we hence abuse notation and treat $\overline{\mathbf{q}}$ as a set of quasi-identifiers $\mathbf{q}\in \mathcal{Q}$. {We assume, as is natural, that for distinct equivalence classes $\overline{\mathbf{q}}, \overline{\mathbf{q}}'$ in any anonymized dataset, we have $\overline{\mathbf{q}}\cap \overline{\mathbf{q}}' = \emptyset$.} Let $\overline{\mathcal{Q}}$ denote the collection of equivalence classes. Further, 
for any $s\in \mathcal{S}$, let $n^{(\overline{\mathbf{q}})}(s):= \sum_{\mathbf{q}\in \overline{\mathbf{q}}} n^{(\mathbf{q})}(s)$. The definition of $\ell$-diversity (and of $(\ell,\delta)$-diversity, in Definition \ref{def:l-delta} in Section \ref{sec:approx}) then extends naturally to anonymized datasets too, by simply replacing $\mathcal{D}$, $\mathbf{q}$, and $\mathcal{Q}$, in Definition \ref{def:l-div}, with $\overline{\mathcal{D}}$, $\overline{\mathbf{q}}$, and $\overline{\mathcal{Q}}$, respectively.

Now, consider an adversary who has access to the anonymized datasets $\left\{\overline{\mathcal{D}}_j: j\in [t]\right\}$, and possesses knowledge of the quasi-identifier $\mathbf{q} = \mathbf{q}_i$ of a specific user $i$ who contributes to all of these datasets. The adversary first identifies the equivalence classes $\overline{\mathbf{q}}_j$ in $\overline{\mathcal{D}}_j$, $j\in [t]$, which contain $\mathbf{q}$. 
Define ${n}^{(\mathbf{q})}_{[t]}(s):=\min\{{n}^{(\overline{\mathbf{q}}_j)}(s): j\in [t]\}$. The adversary then constructs the following \emph{linkage}:
\[
\mathcal{L}^{(\mathbf{q})}_{[t]}:=\{(s,{n}^{(\mathbf{q})}_{[t]}(s)): {n}^{(\mathbf{q})}_{[t]}(s) > 0\}.
\]

Given the above linkage, it is possible to construct the \emph{post-linkage dataset} $\overline{\mathcal{D}}^{(\mathbf{q})}_{[t]}$, which consists of $n^{(\mathbf{q})}_{[t]}(s)$ records for each sensitive attribute $s\in \mathcal{S}$ in the linkage $\mathcal{L}^{(\mathbf{q})}_{[t]}$. The samples in the post-linkage dataset are hence identified by the equivalence class $\bigcap_{j=1}^t \overline{\mathbf{q}}_j$. 

%

In the rest of the paper, our focus will be on the anonymity properties of $\overline{\mathcal{D}}^{(\mathbf{q})}_{[t]}$, ideally requiring that if $\overline{\mathcal{D}}_j$, $j\in [t]$, respect $\ell$-diversity with a moderately large anonymity parameter $\ell$, then $\overline{\mathcal{D}}^{(\mathbf{q})}_{[t]}$ also satisfy the same notion of anonymity, with an anonymity parameter not too degraded.
\section{Worst-Case Degradation of $\ell$-Diversity Upon Linkage}
\label{sec:worst-case}
In this section, we shall argue that the dataset $\overline{\mathcal{D}}^{(\mathbf{q})}_{[t]}$, in the worst-case, respects $\ell$-diversity for a value of $\ell$ that is much smaller than the common anonymity parameter of $\left\{\overline{\mathcal{D}}_j: j\in [t]\right\}$, so long as the alphabet size of the sensitive attributes is at least $3$.

More formally, we are interested in calculating
\begin{equation}
	\label{eq:l-div-worst}
\ell_{[t]}:= \min_{\overline{\mathcal{D}}_1, \ldots, \overline{\mathcal{D}}_t} \{\widehat{\ell}:\ \overline{\mathcal{D}}^{(\mathbf{q})}_{[t]}\text{ obeys $\widehat{\ell}$-diversity}\}.
\end{equation}
{We assume, for simplicity of exposition, that $t$ divides $|\mathcal{S}|-1$; let $L:=|\mathcal{S}|-1$.} In particular, we are interested in the setting where $2\leq t\leq |\mathcal{S}|-1$, and we hence restrict our attention to the case where $|\mathcal{S}|\geq 3$. The following proposition then holds.
\begin{proposition}
	\label{prop:l-div-worst}
	For $t\geq 2$, let $t$ divide $L$. Then, for any $\ell\leq [|\mathcal{S}|]$, we have that $\ell_{[t]} = 1$, if $\ell\leq L\cdot\left(\frac{t-1}{t}\right)+1$, and for $\ell>L\cdot\left(\frac{t-1}{t}\right)+1$, we have
	\[
	\ell_{[t]} = L+1-(L-\ell+1)t.
	\]
\end{proposition}
The proof of Proposition \ref{prop:l-div-worst} proceeds via a straightforward reduction of the problem in \eqref{eq:l-div-worst} to a problem involving the intersections of supports of binary strings of fixed Hamming weight. Indeed, given the equivalence classes $\overline{\mathbf{q}}_j$ in the anonymized datasets $\overline{\mathcal{D}}_j$, one can define the characteristic vector
\[
\iota_j:= \left(\mathds{1}\{n^{(\overline{\mathbf{q}}_j)}(s)>0\}: s\in \mathcal{S}\right),
\]
for $j\in [t]$. Clearly, we have that $w(\iota_j) \geq \ell$, for all $j\in [t]$. Furthermore, if the post-linkage dataset $\overline{\mathcal{D}}_{[t]}$ obeys $\widehat{\ell}$-diversity, then the vector $\iota_{[t]}$, with coordinates
\[
\iota_{[t],s}:=
\begin{cases}
	1,\ \text{if $\iota_{j,s} = 1$, for all $j\in [t]$},\\
	0,\ \text{otherwise},
\end{cases}
\]
is such that $w(\iota_{[t]}) = \widehat{\ell}$. In other words, $\overline{\mathcal{D}}_{[t]}$ obeys $\widehat{\ell}$-diversity, where $\widehat{\ell}$ is the cardinality of the intersection of the supports of the binary strings $\iota_j$, $j\in [t]$. We then have from \eqref{eq:l-div-worst} that
\begin{equation}
	\label{eq:l-div-worst-reduce}
\ell_{[t]}:= \min_{\iota_j, j\in [t]} w(\iota_{[t]}).
\end{equation}
Now, observe that $w(\iota_{[t]})$ cannot decrease if $w(\iota_j)$, for some $j\in [t]$, increases. It hence suffices for us to focus on the setting where $w(\iota_j) = \ell$, for all $j\in [t]$. Further, without loss of generality, assume that $\iota_{1,1} = \ldots = \iota_{t,1} = 1$, since the equivalence classes $\overline{\mathbf{q}}_j$, $j\in [t]$, agree on the contribution of at least one user.

 We define $M\in \{0,1\}^{t\times L}$ as the array whose $j^\text{th}$-row is the vector $(\iota_{j,2},\ldots,\iota_{j,|\mathcal{S}|})$, for $j\in [t]$. Observe that then $\iota_{[t],1} = 1$ and $\iota_{[t],s} = \prod_{j=1}^t M_{j,s}$, for $s\in [2:|\mathcal{S}|]$. The following lemma will be of use to us, in the proof of Proposition \ref{prop:l-div-worst}.
\begin{lemma}
	\label{lem:l-div-worst}
	For $t\geq 2$, let $t$ divide $L$. If $\ell>L\cdot\left(\frac{t-1}{t}\right)+1$, we have that $\ell_{[t]}>1$.
\end{lemma}
\begin{proof}
	It suffices to show that if $\ell>L\cdot\left(\frac{t-1}{t}\right)+1$, the array $M$ cannot have a $0$ in every column. Indeed, $M$ having a $0$ in every column is equivalent to having $\iota_{[t]} = (1,0,0,\ldots,0)$, with $\ell_{[t]} = 1$. Now, note that there are $(\ell-1) t$ $1$s in total in the array $M$, and if there is at least one $0$ in every column of $M$, we must have that $(\ell-1) t\leq L(t-1)$, which is a contradiction.
\end{proof}
Towards proving the proposition, we define the sequences $\mathbf{x}_1,\ldots,$ $\mathbf{x}_t\in \{0,1\}^L$, each of Hamming weight $\ell-1$, such that 
\begin{equation}
	\label{eq:x}
	x_{j,(r)_L} = \begin{cases}
		1,\ \text{if $\frac{(j-1)L}{t}+1\leq r\leq \frac{(j-1)L}{t}+\ell-1$},\\
			0,\ \text{otherwise}.
	\end{cases}
\end{equation}

\begin{proof}[Proof of Proposition \ref{prop:l-div-worst}]
	It can easily be checked by setting $(\iota_{j,2},\ldots,\iota_{j,|\mathcal{S}|}) = \mathbf{x}_j$ as in \eqref{eq:x}, we obtain for $\ell\leq L\cdot \left(\frac{t-1}{t}\right)+1$ that $\ell_{[t]} = 1$. In combination with Lemma \ref{lem:l-div-worst}, it follows that $\ell_{[t]}>1$ if and only if $\ell>L\cdot \left(\frac{t-1}{t}\right)+1$.
	
	Now, consider the case when $\ell>L\cdot \left(\frac{t-1}{t}\right)+1$. From the definition of $\ell_{[t]}$ in \eqref{eq:l-div-worst-reduce}, the number of columns of $M$, all of whose entries equal $1$, is exactly $\ell_{[t]}-1$. Let $M_{\sim\iota_{[t]}}$ denote that submatrix of $M$ consisting of those columns $b\in [L]$ with $\iota_{[t],b+1} = 0$. Clearly, the total number of $0$s in $M_{\sim\iota_{[t]}}$ is $t(L-\ell+1)$. By arguing similar to the proof of Lemma \ref{lem:l-div-worst}, we must then have that 
	\[
	L-\ell_{[t]}+1\leq t(L-\ell+1),
	\]
	or, in other words, that
	\[
	\ell_{[t]}\geq  L+1-(L-\ell+1)t.
	\]
	The proof is then completed by checking that for $(\iota_{j,2},\ldots,\iota_{j,|\mathcal{S}|}) = \mathbf{x}_j$ as above, with $\iota_{j,1} = 1$, for each $j\in [t]$, we have $\ell_{[t]}= L+1-(L-\ell+1)t.$
\end{proof}
As a straightforward corollary of Proposition \ref{prop:l-div-worst}, we obtain the following result when $t=2$:
\begin{corollary}
	\label{cor:l-div-worst}
	Let $L$ be a multiple of $2$. Then, for any $\ell\leq [|\mathcal{S}|]$, we have that $\ell_{[2]} = 1$, if $\ell\leq L/2+1$. For $\ell>L/2+1$, we have 
$
	\ell_{[2]} = 2\ell-L-1.
$
\end{corollary}
Via Corollary \ref{cor:l-div-worst}, we observe as in the setting with $\ell$-diversity that for even moderate values of $\ell$ (less than $L/2$), the post-linkage dataset satisfies only $1$-diversity, in the worst-case. Furthermore, from Proposition \ref{prop:l-div-worst}, we see that as the number $t$ of datasets being linked increases, the threshold $L\cdot \left(\frac{t-1}{t}\right)$ below which all values of $\ell$ result in only $1$-diversity upon linkage in the worst-case, increases. We hence understand that if the datasets $\mathcal{D}_1$ and $\mathcal{D}_2$ are each anonymized in a potentially sample-dependent (or instance-dependent) manner, then the \emph{worst-case} degradation in anonymity upon linkage can be quite significant.

In the next section, we propose the notion of (approximate) $(\ell,\delta)$-diversity, which unlike ``pure'' $\ell$-diversity, only requires that a dataset respect $\ell$-diversity \emph{with high probability}. We then describe a simple mechanism to achieve $(\ell,\delta)$-diversity in a \emph{sample-independent} manner, provided the data samples are independent and identically distributed (i.i.d.). We then argue that approximate diversity thus achieved also results in preservation of the anonymity parameter, with high probability, upon linkage, in sharp contrast to the setting of ``pure'' diversity.
\section{Linkage-Robustness via $(\ell,\delta)$-Diversity}
\label{sec:approx}

In this section, we consider the setting where the samples in datasets $\mathcal{D}_1,\ldots,\mathcal{D}_t$ are independent and identically distributed (i.i.d.), with the samples $(\mathbf{q}_i, s_i)$, $i\in [N]$, in each dataset drawn i.i.d. according to a joint distribution $P_{\mathbf{Q},S} = P_{\mathbf{Q}}\cdot P_{S|\mathbf{Q}}$ (the subscripts may be dropped if the distribution is clear from the context). Such a setting is natural in datasets of medical records, credit card spends, contributions to crowd-funded investments, and the like, where information is \emph{voluntarily} provided by users, thus ensuring independence across the records. We assume, without loss of generality, that $P_\mathbf{Q}(\mathbf{q})>0$, for all $\mathbf{q}\in \mathcal{Q}$, and $P_S(s) = \sum_{\mathbf{q}}P_{\mathbf{Q},S}(\mathbf{q},s)>0$, for all $s\in \mathcal{S}$. As a further useful piece of notation, we let $p_i$, $i\in [|\mathcal{S}|]$, denote the $i^\text{th}$-largest value in $\{P_S(s): s\in \mathcal{S}\}$.

To facilitate the analysis that follows, we first put down the definition of approximate $(\ell,\delta)$-diversity (similar in spirit to the definition in \cite[Def. 4]{prob-anon}); we emphasize that the \emph{definition} does not rely on the i.i.d. nature of the samples, although in the rest of the paper, we shall focus exclusively on this setting.

\begin{definition}[$(\ell,\delta)$-diversity]
	\label{def:l-delta}
	A dataset $\mathcal{D}$ obeys $(\ell,\delta)$-diversity, for some fixed $\delta\in (0,1)$, if it obeys Definition \ref{def:l-div}, with probability at least $1-\delta$, where the probability is calculated over the randomness in sampling $N$ records from the underlying distribution.
\end{definition}

\subsection{A Mechanism for Achieving $(\ell,\delta)$-Diversity}
We now work towards achieving $(\ell,\delta)$-diversity. We assume that there exists a \emph{central} agency (or anonymizer) (CA), for instance, the government, which possesses a high-accuracy estimate of the distribution $P_{\mathbf{Q},S}$, via estimates using datasets of the past. Furthermore, we assume that the CA has exact knowledge of $P_{\mathbf{Q},S}$; any error in its estimates in practice can be modelled by an (additive) error in total variational distance, which can be carried forward in every step of our analysis. 


In such a scenario, there exists a simple, \emph{dataset-independent} generalization strategy that can be conveyed by the CA to the dataset owners, resulting in anonymized datasets that respect approximate anonymity, for sufficiently large number of samples $N$; see Algorithm \ref{alg:l-delta-div} for details. We now elaborate on certain key steps in greater detail. In Step \ref{step:1} of Algorithm \ref{alg:l-delta-div}, the probability mass function $\{P(\overline{\mathbf{q}},s): \overline{\mathbf{q}}\in \overline{\mathcal{Q}}, s\in \mathcal{S}\}$, where $\overline{\mathcal{Q}}$ is the set of (potential) equivalence classes, obeys
$
P(\overline{\mathbf{q}},s) = \sum_{\mathbf{q}\in \overline{\mathbf{q}}} P(\mathbf{q},s),
$
for every $\overline{\mathbf{q}}\in \overline{\mathcal{Q}}, s\in \mathcal{S}$. Furthermore, in Step \ref{step:2}, we define $m:=\min\left\{|\mathcal{Q}|,\frac{1}{\ell p},\frac{\sum_{i=\ell}^{|\mathcal{S}|} p_i}{p}\right\}$. We now claim that there exists at least one choice of equivalence classes $\overline{\mathcal{Q}}$ that respects the constraint in Step 3.
\begin{lemma}
	\label{lem:constraint}
	For any $\ell \in [|\mathcal{S}|]$, for $p\in (0,p_\ell]$, there exists a collection $\overline{\mathcal{Q}}$ such that $|\{s: P(\overline{\mathbf{q}},s)\geq p\}|\geq \ell$, for each $\overline{\mathbf{q}} \in \overline{\mathcal{Q}}$.
\end{lemma}
\begin{proof}
	We claim that $\overline{\mathcal{Q}} = \{\mathcal{Q}\}$ satisfies the constraint in the lemma. Indeed, observe that for $\overline{\mathbf{q}} = \mathcal{Q}$, we have that $P(\overline{\mathbf{q}},s) = P(s)$, for all $s\in \mathcal{S}$. Since $p\leq p_\ell$, it follows that $|\{s: P(s)\geq p\}|\geq \ell$, for any $\ell \in [|\mathcal{S}|]$, by the definition of $p_\ell$.
\end{proof}

We mention that $p_\ell\leq 1/\ell$, for all $\ell \in [|\mathcal{S}|]$. Intuitively, a large number of equivalence classes gives rise to higher ``utility,'' or a smaller ``difference'' between the raw and anonymized datasets, for most reasonable measures of ``utility'' or ``difference'' between datasets. Later, in Section \ref{sec:utility}, we shall specify a generalization algorithm for constructing large numbers of equivalence classes in Step \ref{step:1} of Algorithm \ref{alg:l-delta-div}, for the natural setting where the equivalence classes are constrained to contain \emph{contiguous} quasi-identifiers, according to some fixed total order on the quasi-identifiers. Via a special case of this strategy, we shall also obtain a lower bound on $|\overline{\mathcal{Q}}|$ (see Lemma \ref{lem:lower}), for a broad class of distributions $P_{\mathbf{Q},S}$. 

\begin{algorithm}[t]
	\caption{Achieving $(\ell,\delta)$-diversity}
	\label{alg:l-delta-div}
	\begin{algorithmic}[1]	
		\Procedure{Mech-$(\ell,\delta)$-diversity}{}
		\State Clients convey $\ell\in (0,|\mathcal{S}|]$ and $\delta\in (0,1)$ to the CA.
		\State CA chooses $p\in (0,p_\ell]$ and constructs equivalence classes $\overline{\mathbf{q}}$ such that for each $\overline{\mathbf{q}}$, we have $|\{s: P(\overline{\mathbf{q}},s)\geq p\}|\geq \ell$. \label{step:1}
		\State CA broadcasts the integer {$N =\frac{\ln\left( \frac{m \ell}{\delta}\right)}{\ln\left(\frac{1}{1-p}\right)}$}  and the equivalence classes $\overline{\mathbf{q}}\in \overline{\mathcal{Q}}$ constructed, to the data owners. \label{step:2}
		\State Data owners wait until they accumulate $N$ samples and then construct the anonymized datasets $\overline{\mathcal{D}}_1,\ldots,\overline{\mathcal{D}}_t$.
		\State The datasets $\overline{\mathcal{D}}_1,\ldots,\overline{\mathcal{D}}_t$ are then conveyed to the clients.
		\EndProcedure	
	\end{algorithmic}
\end{algorithm} 

For any equivalence class $\overline{\mathbf{q}}$ obtained via Step~\ref{step:1} in Algorithm~\ref{alg:l-delta-div}, we define $\mathcal{S}_{\overline{\mathbf{q}}}:= \{s: P(\overline{\mathbf{q}},s)\geq p\}$. The following lemma will be useful to us in the proof of anonymity of the datasets returned by Algorithm \ref{alg:l-delta-div}.
\begin{lemma}
	\label{lem:help-hp-l-div}
	We have that $|\overline{\mathcal{Q}}|\leq m$.
\end{lemma}
\begin{proof}
It is clear that $|\overline{\mathcal{Q}}|\leq	|\mathcal{Q}|$. Towards showing that $|\overline{\mathcal{Q}}|\leq\frac{1}{\ell p}$, we observe that
	\begin{align*}
		1&=\sum_{\overline{\mathbf{q}}\in \overline{\mathcal{Q}}} \sum_{s \in \mathcal{S}}P(\overline{\mathbf{q}},s)\\
		&\geq \sum_{\overline{\mathbf{q}}\in \overline{\mathcal{Q}}} \sum_{s \in \mathcal{S}_{\overline{\mathbf{q}}}}P(\overline{\mathbf{q}},s)\\ 
		&\geq \sum_{\overline{\mathbf{q}}\in \overline{\mathcal{Q}}} \ell p = |\overline{\mathcal{Q}}|\cdot\ell p.
	\end{align*}
Next, we prove that $|\overline{\mathcal{Q}}|\leq \frac{\sum_{i=\ell}^{|\mathcal{S}|} p_i}{p}$. Without loss of generality, we assume that $p_i = P_S(i)$, $i\in [|\mathcal{S}|]$. For any fixed $\overline{\mathbf{q}}\in \overline{\mathcal{Q}}$, let $s_{\overline{\mathbf{q}}}$ be the largest index $i\in [|\mathcal{S}|]$ such that $s_i\in S_{\overline{\mathbf{q}}}$. Since $|S_{\overline{\mathbf{q}}}|\geq \ell$, we must have that $s_{\overline{\mathbf{q}}}\in \{\ell,\ell+1,\ldots,|\mathcal{S}|\}$. Now, observe that
	\begin{align*}
	\sum_{i=\ell}^{|\mathcal{S}|} p_i&=\sum_{i=\ell}^{|\mathcal{S}|}\sum_{\overline{\mathbf{q}}\in \overline{\mathcal{Q}}} P(\overline{\mathbf{q}},i)\\
	&= \sum_{\overline{\mathbf{q}}\in \overline{\mathcal{Q}}} \sum_{i=\ell}^{|\mathcal{S}|} P(\overline{\mathbf{q}},i)\\
	&\geq \sum_{\overline{\mathbf{q}}\in \overline{\mathcal{Q}}} P(\overline{\mathbf{q}},s_{\overline{\mathbf{q}}}) \geq  \sum_{\overline{\mathbf{q}}\in \overline{\mathcal{Q}}} p = |\overline{\mathcal{Q}}|\cdot p,
\end{align*}
thus proving that $|\overline{\mathcal{Q}}|\leq \frac{\sum_{i=\ell}^{|\mathcal{S}|} p_i}{p}$.
\end{proof}
\begin{remark}
	We mention that neither $\frac{1}{\ell p}$ nor $\frac{\sum_{i=\ell}^{|\mathcal{S}|} p_i}{p}$ is uniformly smaller than the other, for all marginal distributions $P_S$ and parameters $p\leq p_\ell$. To see why, consider the following two settings with $|\mathcal{S}| = 3$: (i) $p_1 = p_2 = p_3 = \frac{1}{3}$, and (ii) $p_1 = 0.99$, $p_2 = 0.009$, and $p_3 = 0.001$. In the first setting, observe that when $\ell = 2$, we have $\frac{1}{2p} = \frac{1}{\ell p}\leq \frac{2}{3p} = \frac{\sum_{i=\ell}^{|\mathcal{S}|} p_i}{p}$. In the second setting, note that when $\ell = 3$, we have $\frac{0.001}{p} = \frac{\sum_{i=\ell}^{|\mathcal{S}|} p_i}{p} < \frac{1}{\ell p} = \frac{1}{3p}$.
\end{remark}
The following theorem then shows that any dataset $\overline{\mathcal{D}}$ obtained via Algorithm \ref{alg:l-delta-div} obeys $(\ell,\delta)$-diversity. But first, we revisit some notation. For any equivalence class $\overline{\mathbf{q}}$ obtained via Algorithm~\ref{alg:l-delta-div}, let $\mathcal{S}_{\overline{\mathbf{q}}}\big \vert_{\ell}$ denote the $\ell$ sensitive attributes $s$ in $\mathcal{S}_{\overline{\mathbf{q}}}$ with the largest value of $P(\overline{\mathbf{q}},s)$. Furthermore, let $N(\overline{\mathbf{q}},s)$ denote the number of occurrences of records with quasi-identifier $\mathbf{q}$ that lies in $\overline{\mathbf{q}}$, and have sensitive attribute $s\in \mathcal{S}$.

\begin{theorem}
	\label{lem:hp-l-div}
	Any dataset $\overline{\mathcal{D}}$ that is anonymized via {\normalfont\textsc{Mech-$(\ell,\delta)$-Diversity}} satisfies $(\ell,\delta)$-diversity.
\end{theorem}

\begin{proof}
Fix an equivalence class $\overline{\mathbf{q}}$ obtained via {\normalfont\textsc{Mech-$(\ell,\delta)$-Diversity}} (see Algorithm \ref{alg:l-delta-div}). Now, for any $s\in \mathcal{S}_{\overline{\mathbf{q}}}$, we have via the i.i.d. nature of the $N$ samples that
\begin{equation}
	\label{eq:l-div-1}
\Pr\left[N(\overline{\mathbf{q}},s)= 0\right] \leq \left(1-\min_{s\in \mathcal{S}_{\overline{\mathbf{q}}}} P(\overline{\mathbf{q}},s)\right)^N\leq (1-p)^N.
\end{equation}
Let $M_{\overline{\mathbf{q}}}$ be the random variable denoting the number of distinct sensitive attributes in the equivalence class $\overline{\mathbf{q}}$. Clearly,
\begin{align*}
	\Pr\left[M_{\overline{\mathbf{q}}}\geq \ell\right]&\geq \Pr[N(\overline{\mathbf{q}},s) >0,\text{ for all $s\in \mathcal{S}_{\overline{\mathbf{q}}}\big \vert_\ell$}]\\
	&\geq 1-\sum_{s\in \mathcal{S}_{\overline{\mathbf{q}}}\big \vert_\ell} (1-p)^N\\
	&\geq 1-\ell\cdot (1-p)^N,
\end{align*}
where the second inequality follows from \eqref{eq:l-div-1} and an application of a union bound. Finally, via a union bound again, we have
\begin{align*}
	\Pr\left[M_{\overline{\mathbf{q}}}\geq \ell, \text{ for all $\overline{\mathbf{q}}\in \overline{\mathcal{Q}}$}\right]&\geq 1-\ell|\overline{\mathcal{Q}}|\cdot (1-p)^N\\
	&\geq 1-m\ell\cdot (1-p)^N = 1 - \delta,
\end{align*}
where the second inequality above follows from Lemma \ref{lem:help-hp-l-div}. Setting $N =\frac{\ln\left( \frac{m \ell}{\delta}\right)}{\ln\left(\frac{1}{1-p}\right)}$ then gives us the statement of the theorem.
\end{proof}


\begin{remark}
	The choice of the (hyper-)parameter $p$ plays an important role in Algorithm \ref{alg:l-delta-div}.
    For a given value of $\ell$, a relatively small (resp. large) value $p$ implies a potentially high (resp. low) value of $|\overline{\mathcal{Q}}|$ via Step 3 in Algorithm \ref{alg:l-delta-div}. At the same time, a small (resp. large) value of $p$ requires relatively higher (resp. lower) number of samples to be collected by data owners toward ensuring $(\ell, \delta)$-diversity.
    \end{remark}

In the next section, we discuss explicit bounds on the degradation of $(\ell,\delta)$-diversity upon the linkage of datasets.
\subsection{Degradation of $(\ell,\delta)$-Diversity Upon Linkage}
While Theorem \ref{lem:hp-l-div} demonstrates that the datasets $\mathcal{D}_1,\ldots,$ $\mathcal{D}_t$ can be individually anonymized with the aid of the CA, the question that remains to be addressed is the degradation of anonymity upon linkage. The following simple ``composition theorem'' shows that linkage, when carried out after {\normalfont\textsc{Mech-$(\ell,\delta)$-Diversity}} of Algorithm~\ref{alg:l-delta-div}   roughly preserves the strength of the anonymity parameter, at the possible expense of the probability of anonymity failure. We say that a collection of datasets obeys $(\ell,\delta)$-diversity, if the event that \emph{any} dataset in the collection does not respect Definition \ref{def:l-delta} has probability at most $\delta$.
\begin{theorem}
	\label{lem:linkage-l-div}
	The overall post-linkage dataset $\overline{\mathcal{D}}_{[t]}$, defined as $\overline{\mathcal{D}}_{[t]} := \bigcup_{\mathbf{q} \in \mathcal{Q}} \overline{\mathcal{D}}^{(\mathbf{q})}_{[t]} $, where the individual anonymized datasets $\overline{\mathcal{D}}_i$, $1 \le i \le t$, are obtained via {\normalfont\textsc{Mech-$(\ell,\delta)$-Diversity}}, satisfies $(\ell,t\delta)$-diversity.
\end{theorem}
\begin{proof}
	
	Note via Algorithm \ref{alg:l-delta-div} that if quasi-identifiers $\mathbf{q}_1$ and $\mathbf{q}_2$ belong to the same equivalence class, then $\overline{\mathcal{D}}^{(\mathbf{q}_1)}_{[t]} = \overline{\mathcal{D}}^{(\mathbf{q}_2)}_{[t]}$. On the other hand, if $\mathbf{q}_1$ and $\mathbf{q}_2$ do not belong to the same equivalence class, then $\overline{\mathcal{D}}^{(\mathbf{q}_1)}_{[t]} \bigcap \overline{\mathcal{D}}^{(\mathbf{q}_2)}_{[t]} = \emptyset$. For any quasi-identifier $\mathbf{q}$, let $M^{({\mathbf{q}})}_{[t]}$ be the number of distinct sensitive attributes appearing in the post-linkage dataset $\overline{\mathcal{D}}^{(\mathbf{q})}_{[t]}$. It then follows from the proof of Theorem \ref{lem:hp-l-div} and a union bound that
	\[
	\Pr\left[M^{({\mathbf{q}})}_{[t]}\geq \ell,\ \text{for all $\mathbf{q}$} \right]\geq 1-t\delta,
	\]
	thereby yielding the statement of the theorem.
	%
\end{proof}
Via Theorem \ref{lem:linkage-l-div}, we see that {upon linkage of datasets anonymized via the procedure in Algorithm~\ref{alg:l-delta-div}}, the $(\ell,\delta)$-diversity degrades only (linearly) in the probability of anonymity failure, for small $\delta$.

The discussion until now has focussed on mechanisms for achieving (approximate) $(\ell,\delta)$-diversity and an explicit characterization of the degradation of  $(\ell,\delta)$-diversity upon the linkage of datasets, provided that each dataset is independently anonymized via \textsc{Mech-$(\ell,\delta)$-diversity}. As alluded to earlier, a pertinent requirement of the anonymized datasets returned by \textsc{Mech-$(\ell,\delta)$-diversity} is that they give rise to a high ``utility," which, in turn, can be ensured by requiring that the number $|\overline{\mathcal{Q}}|$ of equivalence classes constructed be as large as possible. In the next section, we work towards achieving this objective.
\subsection{Generalization Strategies for High Utility}
\label{sec:utility}
In what follows, we assume that there exists a fixed total order `$<$' on the quasi identifiers $\mathbf{q}\in \mathcal{Q}$, giving rise to $\mathbf{q}_1<\mathbf{q}_2<\ldots<\mathbf{q}_{|\mathcal{Q}|}$.\footnote{Such a total order can be imposed in most settings involving sets of finitely many elements: for instance, when quasi-identifiers are drawn from the set \textsc{Postal-Code}$\times$\textsc{Gender}$\times$\textsc{Age} (see Section \ref{sec:numerics}), one can order the quasi-identifiers according to a natural, lexicographic order.} A natural constraint on the equivalence classes constructed via a generalization strategy is that they comprise \emph{contiguous} quasi-identifiers, with respect to this total order -- we call such a strategy a ``contiguous generalization algorithm". In this section, we first present a contiguous generalization algorithm that can be used in Step \ref{step:1} of Algorithm~\ref{alg:l-delta-div}. Our algorithm, which we call \textsc{Greedy-Generalize}, is shown as Algorithm \ref{alg:greedy}, and is executed by the CA. In words, \textsc{Greedy-Generalize} greedily adds available quasi-identifiers to the current equivalence class $\overline{\mathbf{q}}$, until the constraint that $|S_{\overline{\mathbf{q}}}|\geq \ell$ is satisfied. Note that $R$ denotes the number of equivalence classes $(\overline{\mathbf{q}}_k)_{k\geq 1}$ returned by Algorithm \ref{alg:greedy}. In what follows, we use the term ``admissible generalization algorithm" to refer to any generalization algorithm that returns equivalence classes $\widehat{\overline{\mathbf{q}}}$, each of which satisfies $|S_{\widehat{\overline{\mathbf{q}}}}|\geq \ell$.

\begin{algorithm}[t]
	\caption{A practical contiguous generalization algorithm}
	\label{alg:greedy}
	\begin{algorithmic}[1]	
		\Procedure{Greedy-Generalize}{}
		\State Set $i,j\gets 1$ and $\overline{\mathbf{q}}_1\gets \emptyset$.
		\While{$i<|\mathcal{Q}|$}
		\While{$|S_{\overline{\mathbf{q}}_j}|< \ell$}
		\State Update $\overline{\mathbf{q}}_j\gets \overline{\mathbf{q}}_j\cup \mathbf{q}_i$ and $i\gets i+1$.
		\If{$i = |\mathcal{Q}|+1$} \textbf{break}
		\EndIf
		\EndWhile
		\State Update $j\gets j+1$.
		\EndWhile
		\If{$|S_{\overline{\mathbf{q}}_j}|<\ell$}
		\State Perform $\overline{\mathbf{q}}_{j-1}\gets \overline{\mathbf{q}}_{j-1}\cup \overline{\mathbf{q}}_j$ and set $R\gets j-1$.
		\Else \State Set $R\gets j$.
		\EndIf
		\State Return $(\overline{\mathbf{q}}_k:\ k\in [R])$.
		\EndProcedure	
	\end{algorithmic}
\end{algorithm} 

We now claim that \textsc{Greedy-Generalize} is optimal within the class of contiguous, admissible generalization algorithms, i.e., \textsc{Greedy-Generalize} returns the largest number of equivalence classes among all such algorithms.

\begin{proposition}
	\label{prop:greedy}
	We have that $R\geq R'$, where $R'$ is the number of equivalence classes returned by any admissible, contiguous generalization algorithm.
\end{proposition}
\begin{proof}
	Fix any admissible, contiguous generalization algorithm that returns equivalence classes $\widehat{\overline{\mathbf{q}}}_1,\ldots,\widehat{\overline{\mathbf{q}}}_{R'}$. With some abuse of notation, we treat the ordered quasi-identifiers $(\overline{\mathbf{q}}_j)_{j\leq R}$ (resp. $(\widehat{\overline{\mathbf{q}}}_j)_{j\leq R'}$) as the ordered integers in $[R]$ (resp. $[R']$). For any $k\in [R]$ (resp. $k'\in R'$), let $\mathbf{q}^{(k)}_b$ and $\mathbf{q}^{(k)}_e$ (resp. $\widehat{\mathbf{q}}^{(k')}_b$ and $\widehat{\mathbf{q}}^{(k')}_e$) denote the smallest and largest quasi-identifiers in $\overline{\mathbf{q}}_k$ (resp. $\widehat{\overline{\mathbf{q}}}_{k'}$), respectively. Note that we have $\overline{\mathbf{q}}_b^{(j+1)} = \overline{\mathbf{q}}_e^{(j)}+1$, for all $j\in [R-1]$, and likewise, $\widehat{\overline{\mathbf{q}}}_b^{(j+1)} = \widehat{\overline{\mathbf{q}}}_e^{(j)}+1$, for all $j\in [R'-1]$.
    
    Now, let $t$ be any positive integer such that $\widehat{\mathbf{q}}^{(t)}_e \geq \mathbf{q}^{(t)}_e$. In order to prove the proposition, it suffices then to show that $\widehat{\mathbf{q}}^{(t+1)}_e \geq \mathbf{q}^{(t+1)}_e$. Indeed, it then follows by induction that $\widehat{\mathbf{q}}^{(k)}_e \geq \mathbf{q}^{(k)}_e$, for all $t+1\leq k\leq R'$, thereby implying that $R'\leq R$.
	
	We now prove the above claim. Observe that it suffices to consider the setting where $\mathbf{q}^{(t+1)}_e> \widehat{\mathbf{q}}^{(t+1)}_b$, as otherwise, the claim is trivially satisfied. In such a setting, we claim that it is not possible to have $\widehat{\mathbf{q}}^{(t+1)}_e <\mathbf{q}^{(t+1)}_e$. Indeed, then, we would have that $\{\widehat{\mathbf{q}}^{(t+1)}_b,\ldots,\widehat{\mathbf{q}}^{(t+1)}_e\}\subset \{\mathbf{q}^{(t+1)}_b,\ldots,\mathbf{q}^{(t+1)}_e\}$, which in turn implies that we must have $\mathbf{q}^{(t+1)}_e \leq \widehat{\mathbf{q}}^{(t+1)}_e < \mathbf{q}^{(t+1)}_e$, where the first inequality follows from the definition of $\mathbf{q}^{(t+1)}_e$ via \textsc{Greedy-Generalize}. This results in a contradiction.
\end{proof}
We briefly comment on the time complexity of \textsc{Greedy-Generalize} and hence of Step \ref{step:1} in Algorithm \ref{alg:l-delta-div}. Observe that Algorithm \ref{alg:greedy} conducts a single linear sweep over the $|\mathcal{Q}|$ quasi-identifiers, at each stage performing at most $|\mathcal{S}|$ computations. The overall time complexity of Algorithm \ref{alg:greedy} (for a fixed $p\in (0,p_\ell]$) is hence $O(|\mathcal{Q}||\mathcal{S}|)$, which is efficient for practical implementations. Furthermore, we mention that in Step \ref{step:1} of Algorithm \ref{alg:l-delta-div}, it suffices to vary $p$ over \emph{finitely many} values $p'$ of the form $p'=P(\overline{\mathbf{q}},s)$, for some $\overline{\mathbf{q}}\subseteq \mathcal{Q}$ and $s\in \mathcal{S}$. 

Next, we derive an explicit lower bound on the number $R$ of equivalence classes constructed by \textsc{Greedy-Generalize}, for the case when the random variables $\mathbf{Q}$ and $S$, representing the quasi-identifiers and sensitive attributes, respectively, are independent. Let $\theta:= |\mathcal{Q}|\cdot \min_{\mathbf{q}\in \mathcal{Q}} P_{\mathbf{Q}}(\mathbf{q})$. We derive our lower bound via a contiguous generalization strategy called \textsc{Generalize-Ind}, where all equivalence classes are of the same size. The algorithm \textsc{Generalize-Ind} is shown as Algorithm \ref{alg:generalize-ind}. The following lemma then holds.
\begin{algorithm}[t]
	\caption{Constructing equivalence classes for $P_{\mathbf{Q},S} = P_{\mathbf{Q}}\cdot P_S$}
	\label{alg:generalize-ind}
	\begin{algorithmic}[1]	

		\Procedure{Generalize-Ind}{}
		\State CA fixes $T:= \frac{|\mathcal{Q}|}{\left \lfloor \theta p_\ell/p\right \rfloor}$.
            \State CA constructs the $j^\text{th}$-equivalance class $\overline{\mathbf{q}}_j$, for $j\in [|\mathcal{Q}|/T]$, as:
\begin{equation}
\label{eq:gen-ind}
    \overline{\mathbf{q}}_j = \{\mathbf{q}_{(j-1)T+1},\ldots, \mathbf{q}_{jT}\}.
\end{equation}
	\EndProcedure	
	\end{algorithmic}
\end{algorithm} 


\begin{lemma}
	\label{lem:lower}
	Any equivalence class $\overline{\mathbf{q}}_j$, $j\in [|\mathcal{Q}|/T]$, constructed via \textsc{Generalize-Ind}, obeys $|S_{\overline{\mathbf{q}}_j}|\geq \ell$, 
	when $P_{\mathbf{Q},S} = P_\mathbf{Q}\cdot P_S$. Furthermore, the number of equivalence classes $|\overline{\mathcal{Q}}|$ constructed obeys $|\overline{\mathcal{Q}}|\geq \left \lfloor \frac{\theta p_\ell}{p}\right \rfloor$.
\end{lemma}
\begin{proof}
	The lower bound on the number $|\overline{\mathcal{Q}}|$ of equivalence classes follows immediately from \eqref{eq:gen-ind}. Towards showing that $|\{s: P(\overline{\mathbf{q}},s)\geq p\}|\geq \ell$ holds for all equivalence classes $\overline{\mathbf{q}}_j$, we shall argue that for any $j\in [|\mathcal{Q}|/T]$, we have that all sensitive attributes $i\in [\ell]$ lie in $S_{\overline{\mathbf{q}}_j}$. Indeed, note that for all $i\in [\ell]$,
	\begin{align*}
		P(\overline{\mathbf{q}}_j,i)&= \sum_{r = (j-1)T+1}^{jT} P_{\mathbf{Q},S}(\mathbf{q}_r,i)\\
		&= \sum_{r = (j-1)T+1}^{jT} P_{\mathbf{Q}}(\mathbf{q}_r)P_S(i)\\
		&\geq \sum_{r = (j-1)T+1}^{jT} \frac{\theta}{|\mathcal{Q}|}\cdot P_S(i)\\
		&= \frac{\theta T}{|\mathcal{Q}|}\cdot p_i= \frac{\theta}{\left \lfloor \frac{\theta p_\ell}{p}\right \rfloor} \cdot p_i \geq p,
	\end{align*}
thus showing that $i\in S_{\overline{\mathbf{q}}_j}$, and hence that $|S_{\overline{\mathbf{q}}_j}|\geq \ell$.
\end{proof}
It hence follows that when $P_{\mathbf{Q},S} = P_\mathbf{Q}\cdot P_S$, we have $R\geq \left \lfloor \frac{\theta p_\ell}{p}\right \rfloor$.
\begin{remark}
	\label{remark:1}
	For the special case when $P_\mathbf{Q} = \frac{1}{|\mathcal{Q}|}\cdot \mathbf{1}$, i.e., when the marginal distribution on quasi-identifiers is uniform, we have that the equivalence classes constructed via \textsc{Generalize-Ind} obey $|\overline{\mathcal{Q}}|\geq \left \lfloor \frac{p_\ell}{p}\right \rfloor$. 
\end{remark}
Lemma \ref{lem:lower} yields as a corollary the following explicit characterization of the optimal number of equivalence classes returned by \emph{any} admissible generalization algorithm, for the special case when $P_\mathbf{Q}$ corresponds to a uniform distribution on quasi-identifiers. 
\begin{corollary}
	\label{cor:lower}
	When $P_{\mathbf{Q},S} = P_{\mathbf{Q}}\cdot P_S$ with $P_{\mathbf{Q}} = \frac{1}{|\mathcal{Q}|}\cdot \mathbf{1}$, for any $\ell\in [|\mathcal{S}|]$,  \textsc{Generalize-Ind} returns the largest number of equivalence classes among all admissible generalization algorithms. Furthermore, the number of equivalence classes constructed is $|\overline{\mathcal{Q}}|= \left \lfloor \frac{p_\ell}{p}\right \rfloor$.
\end{corollary}
\begin{proof}
	Via Remark \ref{remark:1}, it suffices to show that any admissible generalization algorithm  returns at most $\frac{p_\ell}{p}$ equivalence classes; it then follows that the lower bound in Lemma \ref{lem:lower} is tight, and is achieved by the \textsc{Generalize-Ind} algorithm. To this end, suppose that for such a generalization algorithm, we have for some equivalence class $\widehat{\overline{\mathbf{q}}}$ it constructs, that $|\widehat{\overline{\mathbf{q}}}|<\frac{|\mathcal{Q}|}{p_\ell/{p}}$. Then, it follows that for any $i\in \{\ell,\ldots,|\mathcal{S}|\}$, we have
	\begin{align*}
		P(\widehat{\overline{\mathbf{q}}},i)&= \sum_{\mathbf{q}\in \widehat{\overline{\mathbf{q}}}} P(\mathbf{q},i)\\
		&= \sum_{\mathbf{q}\in \widehat{\overline{\mathbf{q}}}} \frac{1}{|\mathcal{Q}|}\cdot P_S(i)\\
		&\leq \frac{|\widehat{\overline{\mathbf{q}}}|}{|\mathcal{Q}|}\cdot p_\ell < p,
	\end{align*}
by the assumption that $|\widehat{\overline{\mathbf{q}}}|<\frac{|\mathcal{Q}|}{p_\ell/{p}}$. It then follows that for any such equivalence class $\widehat{\overline{\mathbf{q}}}$, we have $|S_{\widehat{\overline{\mathbf{q}}}}|\leq \ell-1$, which is a contradiction. Thus, every equivalence class $\widehat{\overline{\mathbf{q}}}$ must satisfy $|\widehat{\overline{\mathbf{q}}}|\geq\frac{|\mathcal{Q}|}{p_\ell/{p}}$, implying that the number of equivalence classes is at most $\frac{p_\ell}{p}$.
\end{proof}

\section{Numerical Results}
\label{sec:numerics}
We now present some numerical results pertaining to the number of samples required in Algorithm \ref{alg:l-delta-div}, as functions of the parameters $p$ (that is suitably parameterized with respect with $\ell$) and $\delta$, respectively, in settings of practical interest.

We consider the setting of medical records in hospitals, where the quasi-identifiers are drawn from the set \textsc{Postal-Code}$\times$\textsc{Gender}$\times$\textsc{Age} (with the \textsc{Postal-Code} attribute/set is the collection of neighbourhoods that patients live in) and the sensitive attribute is drawn from the set $\textsc{Disease}$. We choose $|\textsc{Gender}| = 3$, $|\textsc{Age}| = 100$, and $|\textsc{Postal-Code}| = 10$, thereby giving rise to $|\mathcal{Q}| = 3000$, and set $|\textsc{Disease}| = 50$. It is instructive to consider different marginal distributions $P_S$ of interest, which in turn determines the interval $(0,p_\ell]$ that the parameter $p$ lies in. In what follows, for any fixed $\ell\in [|\mathcal{S}|]$, we let $p = \beta \cdot p_\ell$, for some $\beta\in (0,1]$ that is independent of $\ell$. We consider the following marginal distributions $P_S$ of sensitive attributes:
\begin{enumerate}
	\item Uniform marginals: Here, we let $P_S(s) = 1/|\mathcal{S}|$, for all $s\in [|\mathcal{S}|]$; observe that $p_i = 1/|\mathcal{S}|$, for all $i\in [|\mathcal{S}|]$. This results in a ``balanced'' dataset, with respect to the sensitive attributes---a setting of interest in machine learning applications.
	\item Geometric marginals: Here, we let $P_S(s) = p_1 \rho^{s-1}$, for suitable $p_1,\rho\in (0,1)$; observe that $p_i = P_S(i)$, for all $i\in [|\mathcal{S}|]$.  In our experiments, we set $\rho = 0.95$, giving rise to $p_1\approx 0.054$.
\end{enumerate}


Figures \ref{fig:1} and \ref{fig:2} consider the setting of uniform marginals. In Figure \ref{fig:1}, we consider the case when $\beta$ is one of $0.01$ or $0.02$, leading to $p$ being one of $2\times 10^{-4}$ or $4\times 10^{-4}$, respectively, since $p = \beta/|\mathcal{S}|$. For each of these values of $\beta$ (equivalently, $p$), we consider $\ell = 10$ and $\ell = 30$, and plot $N$ as a function of $\delta\in (0,1)$. Figure \ref{fig:2} plots $N$ as a function of $p$, for $\delta$ being one of $0.01$ or $0.001$, and for $\ell$ being one of $\ell = 10$ or $\ell = 30$. Here, our interest is those values of $p$ that are at most $p_\ell = 1/|\mathcal{S}|$, and we let $p$ vary from $10^{-3}$ to $1/|\mathcal{S}|$. Observe that in both settings, we have $\frac{1}{\ell p}   \leq \min\left\{|\mathcal{Q}|,\frac{\sum_{i=\ell}^{|\mathcal{S}|} p_i}{p}\right\} = \min\left\{|\mathcal{Q}|,\frac{|\mathcal{S}|-\ell+1}{|\mathcal{S}|p}\right\}$, resulting in $m\ell = 1/p$ being independent of $\ell$; this fact is reflected in the plots too, in that the plots of the number of samples, for fixed $\beta$ (in Figure \ref{fig:1}) or fixed $\delta$ (in Figure \ref{fig:2}), do not vary with $\ell$. Besides, as expected, the number of samples required decreases as $\delta$ increases, with smaller values of $p$ (intuitively corresponding to larger number of equivalence classes; thereby providing higher ``utility'') requiring more samples for $(\ell,\delta)$-diversity to be satisfied.

\begin{figure}
	\centering
	\subfloat[]{\includegraphics[width = 0.5\textwidth]{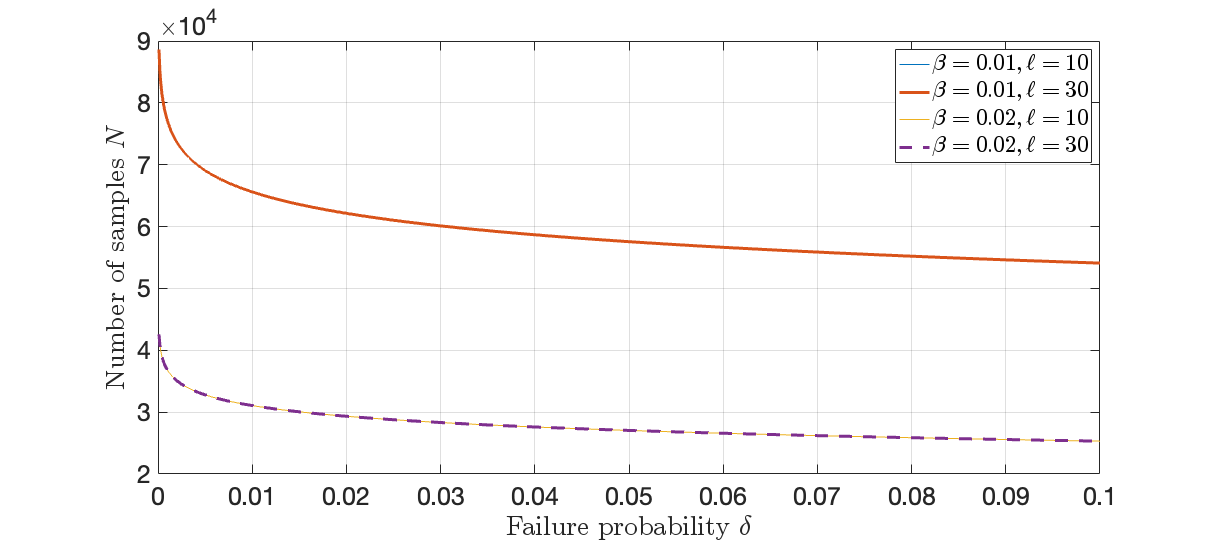} \label{fig:1}}
	\vskip1em
	\subfloat[]{\includegraphics[width = 0.5\textwidth]{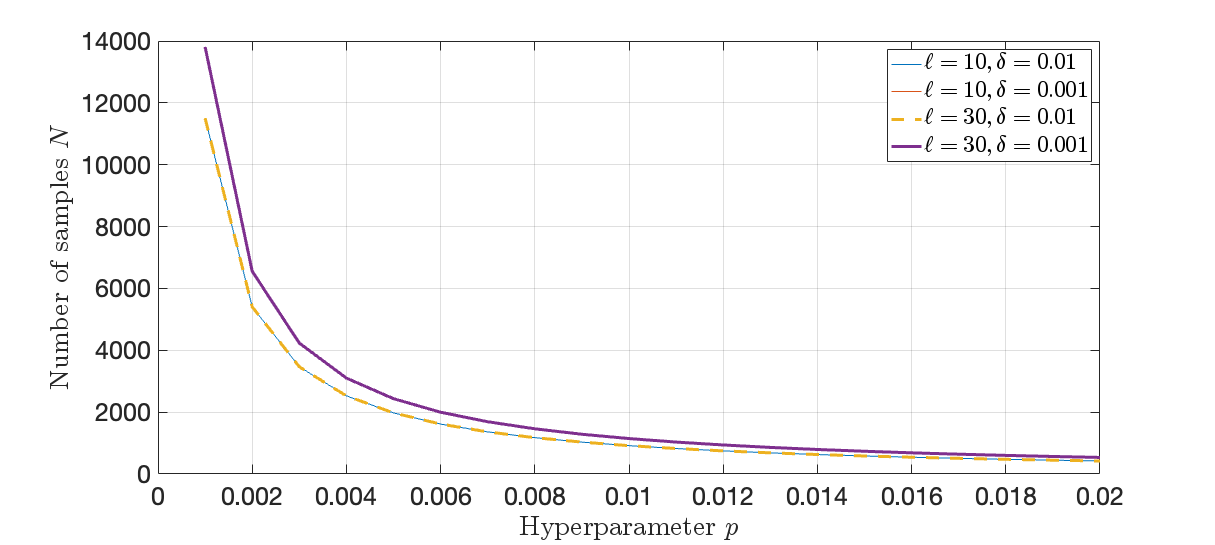} \label{fig:2}}
	\caption{Plots showing the variation of the number of samples $N$ with the failure probability $\delta$ and the diversity parameter $\ell$.}
\end{figure}

\begin{figure}
	\centering
	\subfloat[]{\includegraphics[width = 0.5\textwidth]{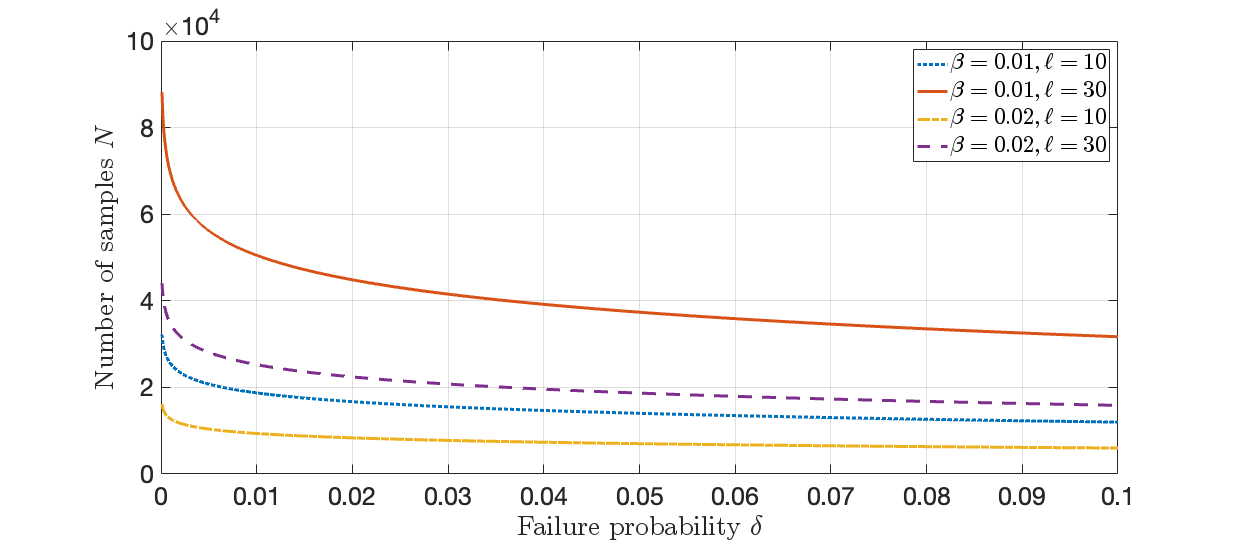} \label{fig:3}}
	\vskip1em
	\subfloat[]{\includegraphics[width = 0.5\textwidth]{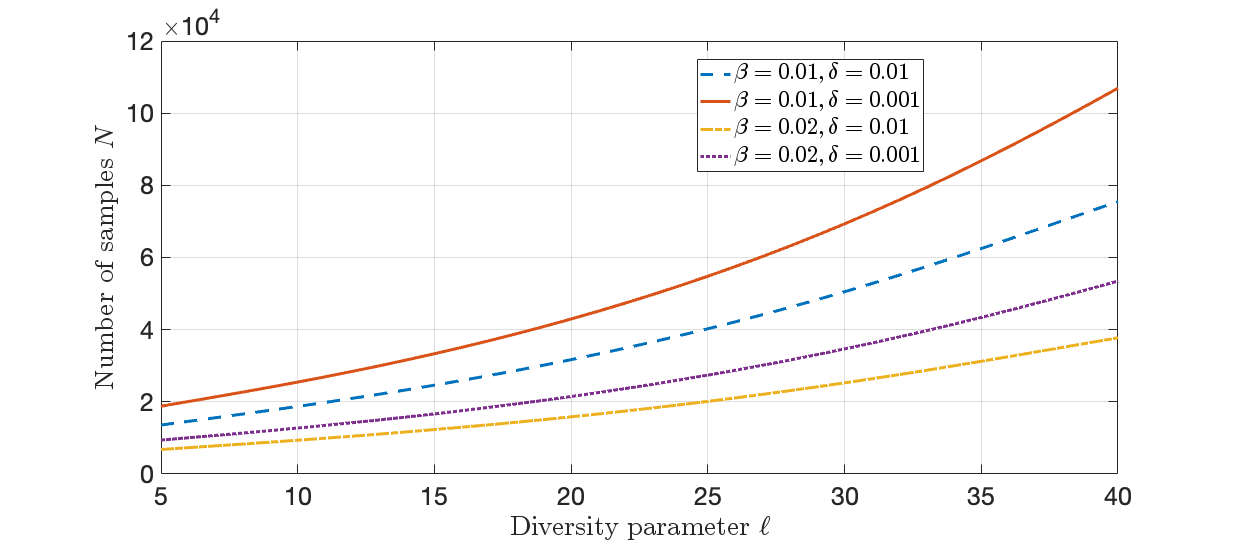} \label{fig:4}}
	\caption{Plots showing the variation of the number of samples $N$ with the failure probability $\delta$ and the diversity parameter $\ell$.}
\end{figure}

Figures  \ref{fig:3} and \ref{fig:4} consider the setting of geometric marginals. In both plots, we consider the case when $\beta$ is one of $0.01$ or $0.02$. For each of these values of $\beta$, Figure \ref{fig:3} considers $\ell = 10$ and $\ell = 30$, and plots $N$ as a function of $\delta\in (0,1)$. Furthermore, for each of the above values of $\beta$, Figure \ref{fig:4} chooses $\delta$ to be one of $0.01$ or $0.001$, and plots $N$ as a function of $\ell\in [5,40]$, by plugging in $p = \beta\cdot p_1 \rho^{\ell-1}$ in the expression for $N$. Again, we note that $N$ decreases with $\delta$, with larger values of $\ell$ requiring more samples for $(\ell,\delta)$-diversity to be satisfied.

We end this section with a couple of remarks. Firstly, we expect that it should be possible to argue that even in the setting where each dataset owner carries out a \emph{dataset-dependent} generalization of quasi-identifiers to achieve ``pure'' anonymity (i.e., $\delta = 0$), the post-linkage dataset will respect the same notion of anonymity, with the anonymity parameter post linkage sufficiently close to that before linkage, \emph{with high probability}, for a large enough number of i.i.d. samples. The intuition behind such a conjecture is that for sufficiently many i.i.d. samples, the fraction of  samples with a given quasi-identifier $\mathbf{q}$ and sensitive attribute $s$ will be close to $P(\mathbf{q},s)$. We leave the design and analysis of such a mechanism for achieving $(\ell,\delta)$-diversity for future work.

Secondly, we mention that it is possible to directly extend the definition of $(\ell,\delta)$-diversity to other notions of anonymity such as $k$-anonymity and entropy $\ell$-diversity \cite{l-div}. In the context of preserving anonymity upon linkages, it is of interest to derive bounds on the sizes of datasets for achieving entropy $\ell$-diversity, via a mechanism similar to Algorithm \ref{alg:l-delta-div}, and analyze the degradation of anonymity upon linkages, in a vein similar to Theorem \ref{lem:linkage-l-div}.

%
\section{Conclusions}
\label{sec:conclusion}
In this paper, we revisited the basic problem of ensuring $\ell$-diversity in the release of entire datasets to clients, for statistical analyses. We considered the setting where an adversary has access to multiple anonymized datasets and wishes to determine the sensitive attribute of a user, whose quasi-identifiers are known, via ``linkage''. We argued that even when each dataset is anonymized to respect ``pure'' $\ell$-diversity, the adersary can, in the worst-case, still exactly extract the sensitive attribute of interest, for even moderately large $\ell$ values. Towards alleviating this issue of anonymity degradation, we introduced the notion of $(\ell,\delta)$-diversity, which only requires that $\ell$-diversity be satisfied with high probability. We then quantified the degradation of $(\ell,\delta)$-diversity under linkage, for datasets of i.i.d. samples, via a simple composition theorem. Finally, we presented an algorithm for maximizing the number of equivalence classes constructed by our algorithm, in the presence of the constraint that each equivalence class contain contiguous quasi-identifiers, and discussed lower bounds on this number, for special sample distributions.

An interesting line of future work will be to extend our analyses to the setting of datasets with correlated samples (or sample distributions with memory, such as the case of voter records) and to other notions of anonymity. Another pertinent question that can be studied is the design of a mechanism for achieving such ``approximate'' anonymity in the absence of a central agency/anonymizer.

\bibliographystyle{ACM-Reference-Format}
\bibliography{references}


\end{document}